\documentclass{SciPost}

\usepackage{hyperref}
\hypersetup{
    colorlinks,
    linkcolor={red!50!black},
    citecolor={blue!50!black},
    urlcolor={blue!80!black}
}

\usepackage[bitstream-charter]{mathdesign}
\DeclareSymbolFont{usualmathcal}{OMS}{cmsy}{m}{n}
\DeclareSymbolFontAlphabet{\mathcal}{usualmathcal}

\usepackage{amsthm}
\usepackage{booktabs}
\usepackage{orcidlink}
\usepackage{CJKutf8}

\newtheorem{theorem}{Theorem}
\newtheorem{lemma}[theorem]{Lemma}
\newtheorem{proposition}[theorem]{Proposition}
\theoremstyle{definition}
\newtheorem{definition}[theorem]{Definition}
\theoremstyle{remark}

\newcommand{\A}{\mathcal{A}}
\newcommand{\R}{\mathbb{R}}
\newcommand{\Q}{\mathbb{Q}}
\newcommand{\ip}[2]{\langle #1,#2\rangle}

\fancypagestyle{SPstyle}{
\fancyhf{}
\lhead{\colorbox{scipostblue}{\bf \color{white} ~SciPost Physics Core }}
\rhead{{\bf \color{scipostdeepblue} ~Submission }}

\fancyfoot[C]{\textbf{\thepage}}
}

\begin{document}

\pagestyle{SPstyle}

\begin{center}{\Large \textbf{\color{scipostdeepblue}{
Strict non-tightness of the level-3 quantum bootstrap for a coupled three-dimensional double well, at an excited eigenstate
}}}\end{center}

\begin{center}\textbf{
Daniel Keren\textsuperscript{1$\star$} and
Ye Zhou (\textmd{\begin{CJK*}{UTF8}{gbsn}周烨\end{CJK*}})\,\orcidlink{0009-0004-7050-7736}\textsuperscript{2}
}\end{center}

\begin{center}
{\bf 1} Department of Computer Science, University of Haifa, Haifa, Israel
\\
{\bf 2} Independent researcher, Kunshan, China
\\[\baselineskip]
$\star$ \href{mailto:dkeren@ds.haifa.ac.il}{\small dkeren@ds.haifa.ac.il}
\end{center}

\section*{\color{scipostdeepblue}{Abstract}}
{\sloppy\textbf{\boldmath{%
For the coupled three-dimensional double well
$H=\tfrac12|p|^2+2\sum_i(q_i^2-1)^2-\sum_{i<j}q_i^2q_j^2$ we exhibit, at the first excited
energy level (and also at the ground state), a linear functional on polynomials of degree
$\le 6$ that satisfies every constraint of the standard level-3 eigenstate bootstrap, has a
\emph{strictly} positive definite moment matrix, and nevertheless assigns a negative value to
a nonnegative polynomial. It is therefore an interior point of the level-3 relaxation that is
realized by no quantum state.
The construction reflects a fact that holds at every level: the eigenstate rows never reach the
top-degree momentum moments. At level 3, where strictly feasible points exist, this means that the
corresponding block of the moment matrix carries no information about which energies admit one. Every algebraic ingredient is an
exact rational certificate, and the only numerical ingredients (spectral enclosures of an
eigenfunction) are carried out in rigorous ball arithmetic.
}}\par}

\vspace{\baselineskip}

\noindent\textcolor{white!90!black}{%
\fbox{\parbox{0.975\linewidth}{%
\textcolor{white!40!black}{\begin{tabular}{lr}%
  \begin{minipage}{0.6\textwidth}%
    {\small Copyright attribution to authors. \newline
    This work is a submission to SciPost Physics Core. \newline
    License information to appear upon publication. \newline
    Publication information to appear upon publication.}
  \end{minipage} & \begin{minipage}{0.4\textwidth}
    {\small Received Date \newline Accepted Date \newline Published Date}%
  \end{minipage}
\end{tabular}}
}}
}


\vspace{10pt}
\noindent\rule{\textwidth}{1pt}
\tableofcontents
\noindent\rule{\textwidth}{1pt}
\vspace{10pt}

\section{Introduction}

The quantum-mechanical bootstrap of Han, Hartnoll and Kruthoff~\cite{HHK} replaces the
eigenvalue problem $H\psi=E\psi$ by conditions on the expectation values
$\langle O\rangle=\ip{\psi}{O\psi}$. These are the linear relations $\langle[H,O]\rangle=0$ and
$\langle O(H-E)\rangle=\langle(H-E)O\rangle=0$, together with positive semidefiniteness of a finite moment
matrix built from operators up to a size cutoff. Increasing the cutoff level gives a hierarchy of more stringent
relaxations. The bootstrap has been applied to a variety of one- and few-body systems
(e.g.~\cite{BH1,BH2,BDDJK}), and in practice the allowed region shrinks rapidly as the
cutoff level grows. However, what a fixed level guarantees is less clear. A point of the level-$K$ feasible
set satisfies every available relation at that level, but it need not be the moment vector of
any state. If such spurious points exist only on the boundary of the feasible set, where the
moment matrix is singular, they might be dismissed as artifacts of truncation. This note
shows otherwise: for a natural polynomial Hamiltonian, at an excited
eigenstate, the level-3 feasible set contains a point with strictly positive definite moment
matrix that is realized by no state.
Appendix~\ref{sec:corner} shows that the mechanism is
structural: at every level, the eigenstate rows are blind to the top-degree momentum block of the moment
matrix, and at levels up to 3 this block carries no information about which energies admit a strictly
feasible point.

\paragraph{Related work.}
Moment-based quantization predates the modern operator bootstrap. The
Eigenvalue Moment Method of Handy and Bessis uses classical moment-positivity
constraints to obtain bounds on Schr\"odinger eigenenergies~\cite{Handy1985};
an extension based on $|\Psi|^2$ also treats excited states~\cite{Handy1987}.
This line of work includes multidimensional spectral applications~\cite{Handy1988PRL}
and a linear-programming formulation of its convex moment constraints~\cite{Handy1988PRA}.
The moment method was also extended to non-Hermitian systems, including regimes with complex discrete eigenenergies~\cite{Handy2001PTSymmetry}.
Our focus instead is finite-order representability of noncommutative
operator moments: strict moment-matrix positivity does not guarantee that
a feasible truncated functional is physical.
Relaxations of this kind are related to moment hierarchies of Lasserre type~\cite{Lasserre}. Their noncommutative
analogues converge under an Archimedean condition~\cite{PNA}, which fails for the unbounded operators $q$
and $p$. Convergence of the quantum-mechanical bootstrap was posed as an open
problem in~\cite{LM}; spurious allowed regions have been observed numerically~\cite{BDDJK}. The obstruction
used here is the classical gap between nonnegative polynomials and sums of squares~\cite{Hilbert,Blekherman}.
A spurious point of the weaker stationary relaxation, which omits the eigenstate rows, was given
in~\cite{K26}. The special role of the momentum moments has also been noted for the Kronig--Penney
model~\cite{BBB22}. In a companion paper, Zhou and Keren~\cite{ZK26} prove that bounds on a fixed observable
remain nonexact at every finite order for a quasi-exactly solvable sextic oscillator, even under finite
polynomial strengthenings. They also prove convergence for polynomial Schr\"odinger operators under
explicit sum-of-squares confinement and derivative hypotheses. The present paper gives a coupled
three-dimensional example with a strictly positive definite spurious moment matrix at level~3.

\paragraph{The role of numerics.}
Everything in the construction is exact except the moment data of an eigenfunction, which has no closed
form. Two of its properties are needed quantitatively: that its moment matrix is strictly positive
definite, and an upper bound on its value on a fixed nonnegative polynomial, the \emph{witness} $f$ of
Section~\ref{sec:construction}. They are proved in Section~\ref{sec:proofB} by a chain of standard
validated-numerics tools: a certified lower bound on the next eigenvalue in the symmetry sector (a separable
comparison operator and a truncation argument in a Hermite basis); Temple's inequality~\cite{Temple}, which
encloses the energy and the angle between the true eigenfunction and a computed one; an explicit
relative-boundedness estimate (Lemma~\ref{lem:cube}), which turns control of the energy residual into control
of the error in every moment entering the moment matrix, derivatives included; and ball
arithmetic~\cite{Arb} for every computation with the approximate eigenfunction, so that every number used is
a guaranteed enclosure. The witness $f$ and the functional $\ell$ that pairs with it are found by numerical
optimization, but the optimization plays no role in the proof: $f$, $\ell$ and all algebraic certificates are
exact rational objects, verified in exact arithmetic. Section~\ref{sec:repro} lists the accompanying code and certificate files, and describes independent
recomputations of the key quantities.

\section{Setting}

\paragraph{Weyl algebra.}
Let $\A$ be the Weyl algebra generated by $q_1,q_2,q_3,p_1,p_2,p_3$ with
$[q_j,p_k]=i\delta_{jk}$, realized on $L^2(\R^3)$ by $q_j=x_j$, $p_j=-i\partial_j$.
Every element has a unique standard-ordered expansion $\sum c_{\alpha\beta}\,q^\alpha p^\beta$
(all $q$'s to the left); $\A_{\le d}$ denotes the span of the monomials $q^\alpha p^\beta$
with $|\alpha|+|\beta|\le d$. For $A\in\A_{\le 6}$ we write $\pi_6(A)$ for the homogeneous
polynomial $\sum_{|\beta|=6}c_{0\beta}\,p^\beta$ formed by the pure-momentum monomials of top
degree in the standard-ordered expansion.

\paragraph{Hamiltonian.}
\begin{equation}\label{eq:H}
H=\tfrac12|p|^2+V(q),\qquad V(q)=2\sum_{i=1}^3(q_i^2-1)^2-\sum_{i<j}q_i^2q_j^2 .
\end{equation}
$V$ is bounded from below, $H$ has compact resolvent and purely discrete spectrum, and its
eigenfunctions are Schwartz functions. $H$ commutes with the reflections $q_i\mapsto -q_i$
and with permutations of the axes.
The relaxation replaces the expectation values $\langle O\rangle$ by the values $L(O)$ of a linear
functional $L$, required only to satisfy the conditions that expectation values in an eigenstate would.
\begin{definition}[Level-$K$ eigenstate bootstrap]\label{def:boot}
Fix $E\in\R$. A linear functional $L:\A_{\le 2K}\to\mathbb{C}$ is \emph{feasible at level $K$}
if
\begin{enumerate}
\item $L(1)=1$ and $L(A^\dagger)=\overline{L(A)}$ (Hermiticity);
\item the moment matrix $M_K(L)=\big(L(m_a^\dagger m_b)\big)_{a,b}$, indexed by the
monomials $m_a=q^\alpha p^\beta$ with $|\alpha|+|\beta|\le K$, is positive semidefinite;
\item (eigenstate rows) $L\big(O(H-E)\big)=L\big((H-E)O\big)=0$ for every monomial $O$ with
$\deg O\le 2K-4$;
\item (stationarity rows) $L([H,O])=0$ for every monomial $O$ with $[H,O]\in\A_{\le 2K}$.
\end{enumerate}
\end{definition}

This is the standard form of the constraints~\cite{HHK,BH1,BDDJK}: every row is generated
by a single monomial $O$ of bounded degree. Section~\ref{sec:discussion} discusses the larger
set of algebraic consequences obtained from linear combinations of generators of higher
degree that cancel back into $\A_{\le 2K}$. Throughout, interior is understood relative to the affine
space defined by normalization, Hermiticity and the linear bootstrap relations.

For $K=3$ the moment matrix is $84\times 84$. If $H\psi=E\psi$ and $\|\psi\|=1$, the
\emph{base functional} $B(A)=\ip{\psi}{A\psi}$ is feasible at every level: it is the moment data
the bootstrap is meant to recover. Theorem~\ref{thm:main} departs
from such a $B$ along a
direction invisible to the constraints. A feasible $L$ is \emph{spurious} if it is not of the form $A\mapsto\operatorname{tr}(\rho A)$ for any
density operator $\rho$ (positive, trace class, $\operatorname{tr}\rho=1$) for which these traces are defined. A
sufficient condition is $L(f)<0$ for some $f$ that is nonnegative as an operator; such an $f$ is
called a \emph{witness} (to the spuriousness of $L$).

\section{Main result}

\begin{theorem}\label{thm:main}
Let $E_1$ be the first excited eigenvalue of \eqref{eq:H}. It has multiplicity  three
(Proposition~\ref{prop:spec}) and
\[
E_1\in[\,1.61488137077986920\pm 2\cdot10^{-17}\,],\qquad
E_0\le 1.53482611032904480 .
\]
Let $\psi$ be the normalized eigenfunction of $E_1$ that is odd in $q_1$, even in $q_2,q_3$
and symmetric under $q_2\leftrightarrow q_3$, and let $B(A)=\ip{\psi}{A\psi}$. Let $\Delta$,
$\ell$ and $f$ be as in Section~\ref{sec:construction}. Then for every $\tau>10683.13$ the
functional
\[
L_\tau=B+\tau\Delta
\]
is feasible at level $K=3$ with $E=E_1$ (Definition~\ref{def:boot}), its moment matrix is
strictly positive definite,
\[
\lambda_{\min}\big(M_3(L_\tau)\big)\ \ge\ 1.45\cdot10^{-5},
\]
and $L_\tau(f)<0$ for the polynomial $f$ of Section~\ref{sec:construction}, which satisfies
$f(\xi)\ge 0$ for all $\xi\in\R^3$. Hence $L_\tau$ is an interior spurious point of the level-3
eigenstate relaxation.

The same statement holds for the ground state $\psi_0$ (energy $E_0$, fully symmetric sector)
with $\tau>10079.54$ and $\lambda_{\min}(M_3(L_\tau))\ge 1.41\cdot10^{-5}$.
\end{theorem}

The next proposition identifies the level: it shows that $E_1$ is the first excited eigenvalue and has
multiplicity three, so that $\psi$ is well defined. Its proof uses the tools of Section~\ref{sec:proofB}
and is given at the end of that section.
\begin{proposition}\label{prop:spec}
The eigenvalues of \eqref{eq:H} below $1.9720$ are exactly the following, counted with
multiplicity:
\[
\setlength{\arraycolsep}{1.5pt}
\begin{array}{lll}
E_0\in[1.53482611032697,\,1.53482611032905] & \text{mult.\ }1 & \text{parity class }(+,+,+)\\
E_1\in[1.61488137077792,\,1.61488137077988] & \text{mult.\ }3 & \text{classes }(-,+,+),(+,-,+),(+,+,-)\\
E_2\in[1.68248592083046,\,1.68248592083215] & \text{mult.\ }3 & \text{classes }(-,-,+),(-,+,-),(+,-,-)\\
E_3\in[1.74146650963446,\,1.74146650963490] & \text{mult.\ }1 & \text{class }(-,-,-)
\end{array}
\]
Here the class records the parity of the eigenfunction under $q_i\mapsto-q_i$. These eight
levels form the tunneling multiplet of the eight wells at $(\pm\sqrt2,\pm\sqrt2,\pm\sqrt2)$,
the global minima of $V$ (where $V=-6$).
\end{proposition}

The first excited level is threefold degenerate. The eigenfunction $\psi$ of Theorem~\ref{thm:main} is one
real eigenfunction in it, of definite parity; its images under permutations of the axes are covered by the
same argument after relabeling.

We point out three features of the result. First, the Hamiltonian is a polynomial,
cubic-symmetric coupled double well. The construction uses neither quasi-exact solvability nor
auxiliary degrees of freedom. Second, the constraint rows are imposed at the true eigenvalue, for which no closed form is
available and which is known only through a rigorous enclosure; this is possible because the direction
added to $B$ (Section~\ref{sec:construction}) satisfies the 
constraints for every $E$. Third,
the moment matrix is strictly positive definite, so the point lies in the interior of the feasible set.
At every level $K\ge4$, the coefficient vector of the degree-4 polynomial $H-E$ lies in the kernel of
$M_K$ for every feasible functional (Appendix~\ref{rem:higher}).

The proof consists of two tiers. Lemmas~\ref{lem:witness}--\ref{lem:MDelta} of Section~\ref{sec:construction} are
exact and use no numerical analysis; together with the mere existence of an eigenstate, they already show
that the level-3 eigenstate relaxation is not tight, at every eigenvalue (Appendix~\ref{rem:higher}). The
numerical analysis of Section~\ref{sec:proofB} is needed only for the two quantitative claims:
strict positive definiteness of $M_3(L_\tau)$, and the explicit threshold for $\tau$.

\section{The construction}\label{sec:construction}

\paragraph{The witness $f$ and the functional $\ell$.}
The construction rests on a pair of objects in the momenta alone: a nonnegative sextic $f$ that is not
a sum of squares, and a linear functional $\ell$ that is nonnegative on every sum of
squares but negative on $f$. The pair certifies the gap between nonnegative polynomials and sums of squares in the one place the
constraint rows cannot see (Lemma~\ref{lem:rows}). Both $f,\ell$ are exact rational objects, listed in
Table~\ref{tab:fl}. The witness is a
homogeneous sextic in the momenta that depends only on their squares: $f(p)=\sum_e f_e\,p^e$, with $e$
ranging over the ten exponent vectors with even entries summing to~$6$ and
$p^e=p_1^{e_1}p_2^{e_2}p_3^{e_3}$. It is the rational rounding of the cubic found by optimization
(see below), evaluated at $(p_1^2,p_2^2,p_3^2)$, plus $|p|^6/270000$, which makes the certificate of
Lemma~\ref{lem:witness}(1) strictly interior. The functional $\ell$ is defined on degree-6 momentum
monomials, with $\ell(p^e)=0$ unless every entry of $e$ is even; its values on the ten even monomials
are also listed in Table~\ref{tab:fl}.

\begin{table}[htb]\centering
\renewcommand{\arraystretch}{1.15}
\begin{tabular}{ccc}\toprule
$e$ & $f_e$ & $\ell(p^e)$\\\midrule
$(6,0,0)$ & $119045273/1350000000$ & $6838869/50000000$\\
$(0,6,0)$ & $119072651/1350000000$ & $683773/5000000$\\
$(0,0,6)$ & $11906231/135000000$ & $2735249/20000000$\\
$(4,2,0)$ & $-149982839/900000000$ & $992129/10000000$\\
$(0,4,2)$ & $-18750187/112500000$ & $2480133/25000000$\\
$(2,0,4)$ & $-18748207/112500000$ & $4960551/50000000$\\
$(2,4,0)$ & $1/90000$ & $837787/10000000$\\
$(0,2,4)$ & $1/90000$ & $8377993/100000000$\\
$(4,0,2)$ & $1/90000$ & $523663/6250000$\\
$(2,2,2)$ & $52977461/225000000$ & $2036581/50000000$\\\bottomrule
\end{tabular}
\caption{The witness $f(p)=\sum_e f_e\,p^e$ and the functional $\ell$ on the even sextic
momentum monomials ($\ell(p^e)=0$ for all other $e$).}\label{tab:fl}
\end{table}

The proofs of Lemmas~\ref{lem:witness}--\ref{lem:MDelta} are computations in exact rational arithmetic:
every identity is checked coefficient by coefficient, and every positive-definiteness claim by an exact
$LDL^\top$ factorization, so no floating-point number enters. Each was verified by two independent
implementations (Section~\ref{sec:repro}).

\begin{lemma} \label{lem:witness}
\begin{enumerate}
\item $|p|^2f(p)=v(p)^{\!\top}G\,v(p)$, where $v$ is the vector of the 15 degree-4 monomials in
$p$ and $G$ is a rational symmetric positive definite matrix (proved via exact $LDL^\top$ decomposition;
Appendix~\ref{app:G}). Hence $f(\xi)>0$ for all $\xi\in\R^3\setminus\{0\}$. Since the momenta
commute, $f(p)$ acts in the momentum representation as multiplication by $f(\xi)$, so $f(p)$ is a
nonnegative operator.
\item The Gram matrix $\big(\ell(p^{\beta+\beta'})\big)$ over the ten degree-3 momentum
monomials is positive definite (exact $LDL^\top$; smallest pivot $3.37\cdot10^{-3}$).
\item $\ell(f)=-11470021230637/(3\cdot10^{15})\approx-3.8233\cdot10^{-3}$.
\end{enumerate}
\end{lemma}

\paragraph{The ``tail functional'' $\Delta$.}
\begin{equation}\label{eq:Delta}
\Delta(A)=\ell\big(\pi_6(A)\big),\qquad A\in\A_{\le 6}.
\end{equation}
Plainly put, $\Delta$ reads only the coefficients of the top-degree, pure-momentum monomials
$p^\beta$ ($|\beta|=6$) of $A$ and applies $\ell$ to them. It vanishes on $\A_{\le 5}$; in
particular $\Delta(1)=0$. On a pure-momentum sextic such as $f$ it coincides with $\ell$,
$\Delta(f)=\ell(f)$.

\begin{lemma} \label{lem:rows}
$\Delta$ satisfies Hermiticity and annihilates every eigenstate row and every stationarity row
of Definition~\ref{def:boot} at $K=3$, for every $E\in\R$. Concretely, for all 28 monomials
$O$ of degree $\le 2$ we have $\Delta(OH)=\Delta(HO)=\Delta(O)=0$. For all 259 monomials
$O$ with $[H,O]\in\A_{\le6}$ we have $\Delta([H,O])=0$; these are the pure-position monomials of
degree $\le6$ and all other monomials of degree $\le4$ (for any other monomial $O$ the
standard-ordered expansion of $[H,O]$ has degree $\deg O+2$, or $\deg O$ if $O$ is
pure-position). For all 924
monomials $A$ of degree $\le 6$ we have $\Delta(A^\dagger)=\overline{\Delta(A)}$.
\end{lemma}

The reason is structural. The only momentum dependence of $H$ is $\tfrac12|p|^2$, and $V$ is a
pure-position polynomial. A level-3 row therefore never produces a pure-momentum term of top
degree $6$ that $\ell$ could see. 

\begin{lemma}\label{lem:MDelta}
$M_3(\Delta)$ vanishes outside the $10\times10$ block whose rows and columns are indexed by the ten
pure-momentum monomials $p^\beta$, $|\beta|=3$; on that block it equals the positive definite Gram matrix
$G_\ell$ of Lemma~\ref{lem:witness}(2). In particular $M_3(\Delta)\succeq 0$.
\end{lemma}

\paragraph{Proof of Theorem~\ref{thm:main}, given Proposition~\ref{prop:B}.}
$L_\tau(1)=B(1)=1$. Hermiticity and all rows hold for $B$, because $\psi$ is an eigenfunction, and
for $\Delta$ by Lemma~\ref{lem:rows}; since the conditions are linear, they hold for $L_\tau$. By
Lemma~\ref{lem:MDelta},
$M_3(L_\tau)=M_3(B)+\tau M_3(\Delta)\succeq M_3(B)\succ 0$, with the bound of
Proposition~\ref{prop:B}. Finally,
$L_\tau(f)=B(f)+\tau\,\ell(f)<0$ when  $\tau>B(f)/|\ell(f)|$. With the upper bound on $B(f)$
from Proposition~\ref{prop:B} this holds for $\tau>10683.13$, and for $\tau>10079.54$ in the
ground-state case. \qed

\subsection*{How $f$ and $\ell$ were found: the alternation method}

This subsection describes how $f$ and $\ell$ were obtained. The proof of Theorem~\ref{thm:main} does not
depend on this: it uses only the exact checks of Lemmas~\ref{lem:witness}--\ref{lem:MDelta}, which any pair
$(f,\ell)$ passing them would satisfy equally well. The method is needed to find such a pair in the first
place, at this and at higher levels (Appendix~\ref{rem:higher}), and to reproduce the certificates.

Write the witness as $f(p)=F(p_1^2,p_2^2,p_3^2)$ with $F$ a cubic form in $x=(x_1,x_2,x_3)$,
and view $\ell$ as a vector of values on the even sextic momentum monomials. The smallest value of $\ell(f)$ over admissible pairs (with $f$ nonnegative, $\ell$ nonnegative on
sums of squares of cubics, and both normalized) is the value of the bilinear problem
\begin{equation}\label{eq:bilinear}
\min_{\ell,\,F}\ \ell\big(F(p^2)\big)\quad\text{subject to}\quad
G_\ell\succeq 0,\ \operatorname{tr}G_\ell=1,\qquad F\ge0\ \text{on }\R^3_{\ge0},\ \|F\|_1\le1 .
\end{equation}
The first constraint says that $\ell$ is nonnegative on every sum of squares of cubics, since
$\ell(g^2)=c^\top G_\ell\,c\ge0$ for every cubic $g=\sum_\beta c_\beta p^\beta$ (that is, $\ell$ lies in the cone
dual to the sums of squares); the second says that $F$ is copositive, which makes $f=F(p_1^2,p_2^2,p_3^2)$
nonnegative everywhere. The objective is linear in each variable separately, and the two constraints pull
in opposite directions: a negative value $\ell(f)<0$ is possible only if $f$ is nonnegative but \emph{not}
a sum of squares of cubics. The value of~\eqref{eq:bilinear} measures how strongly a normalized $\ell$ can
detect such an $f$. We solve it by alternating minimization:
\begin{enumerate}
\item[(a)] \emph{SDP step.} For fixed $F$, minimize $\ell(F(p^2))$ over $\ell$ with
$G_\ell\succeq0$, $\operatorname{tr}G_\ell=1$.
\item[(b)] \emph{LP step with cutting planes.} For fixed $\ell$, minimize the linear function
$F\mapsto\ell(F(p^2))$ over $\|F\|_1\le1$, subject to $F(x)\ge0$ at a grid of sample points $x$
of the simplex. After each LP solve, local minimization (Nelder-Mead, started from the
lowest grid values) searches for a point where $F<0$. If one is found it is added as a cut
and the LP is re-solved; otherwise the step ends. This search does not certify copositivity; for the final
witness, nonnegativity is proved exactly (Lemma~\ref{lem:witness}).
\end{enumerate}
Exact minimization in each block of~\eqref{eq:bilinear}, starting from a feasible pair, would give a
nonincreasing objective. The cutting-plane implementation uses a local search for violated copositivity
constraints, so this monotonicity is not guaranteed. No global optimality claim is needed: the final
rational pair is checked independently in exact arithmetic.

In practice, the Robinson form $\sum x_i^3-\sum_{i\ne j}x_i^2x_j+3x_1x_2x_3$~\cite{Robinson}, which after
the substitution $x_i=p_i^2$ is the classical nonnegative ternary sextic that is not a sum of squares, is
the natural seed. It is, however, a stationary point of the alternation (started there, it stays at the
value $-4.0333\cdot10^{-3}$), so the seed is perturbed. Run to convergence, the alternation reaches
$-4.7723\cdot10^{-3}$. This value does not depend on $V$, since the tail construction never sees it. The
final $F$ is a cyclic cubic $a\sum x_i^3-b\,(x_1^2x_2+x_2^2x_3+x_3^2x_1)+c\,x_1x_2x_3$ with
$(a,b,c)\approx(0.0882,0.1667,0.2354)$. The numerical search finds values within about $10^{-10}$ of zero
at the apparent minimum on the simplex. This does not certify copositivity or global optimality.
The proof uses a nearby rational pair whose required properties are verified exactly, as follows.

To pass to exact data we perturb  both objects slightly into the interior of their cones: we replace $F(p^2)$
by $f=F(p^2)+|p|^6/270000$, which makes the Gram matrix of Lemma~\ref{lem:witness}(1) positive definite,
and $\ell$ by $0.98\,\ell+0.02\,\ell_{\rm Gauss}$ (the Gaussian moments), which makes $G_\ell$ positive
definite. We then round to rationals and project exactly onto the linear identities. This costs part of
the violation: the certified value is $\ell(f)=-3.8233\cdot10^{-3}$, compared with the numerical value
$-4.7723\cdot10^{-3}$.

The proof of Theorem~\ref{thm:main} above used two properties of the base functional
$B(A)=\langle\psi,A\psi\rangle$: that $M_3(B)$ is positive definite, with an explicit lower bound on its smallest
eigenvalue, and an upper bound on $B(f)$, which fixes the threshold for $\tau$. Both concern an eigenfunction
with no closed form, and they are proved in Section~\ref{sec:proofB} by rigorous numerical analysis.

\begin{proposition} \label{prop:B}
For the state $\psi$ of Theorem~\ref{thm:main},
$\lambda_{\min}(M_3(B))\ge 1.4510\cdot10^{-5}$ and
$B(f)\le 40.84521$.
For the ground state, $\lambda_{\min}(M_3(B))\ge 1.4135\cdot10^{-5}$ and
$B(f)\le 38.53751$.
\end{proposition}

\section{Proof of Proposition~\ref{prop:B}}\label{sec:proofB}

All numbers below are rigorous. They are computed in \emph{ball arithmetic}: every real
number is carried as a midpoint together with a guaranteed error radius, and every operation
returns a ball that provably contains the exact result, including all rounding errors. We use
Arb~\cite{Arb} through python-flint, at 200-bit precision. Irrational matrix elements such as $\sqrt n$ (which appear in the Hermite basis)
are therefore enclosed, never merely approximated.

Let $\mathcal S$ be the symmetry sector (odd in $q_1$, even in $q_2,q_3$, symmetric under
$q_2\leftrightarrow q_3$). $H$ leaves $\mathcal S$ invariant and $\psi$ is its lowest
eigenfunction there.
Write $T=H+c$ with $c=5$, so that $T\ge E_0+5>0$; the same constant appears in Lemma~\ref{lem:cube}.

\paragraph{Step 1: a gap in the sector.}
$H-H_{\rm sep}=\tfrac12\sum_{i<j}(q_i^2-q_j^2)^2\ge 0$, where
$H_{\rm sep}=\sum_i \tilde h(q_i,p_i)$ and $\tilde h=\tfrac12p^2+q^4-4q^2+2$. Both operators
preserve $\mathcal S$, so the min-max principle (Courant--Fischer) gives $\lambda_k(H|_{\mathcal S})\ge\lambda_k(H_{\rm sep}|_{\mathcal S})$.
The two lowest even and odd levels of $\tilde h$ are bounded below by a truncation argument in
the Hermite basis of $\tfrac12p^2+\tfrac12w^2q^2$ ($w=4$). On the complement $Q$ of the first
$N_1=20$ basis functions of a given parity one has $Q\tilde hQ\ge\alpha Q$, with
$\alpha=w(2N_1+\text{parity}+\tfrac12)-c_1$ and
$c_1=34$; here $-c_1$ is the minimum of the multiplication operator
$\tilde h-(\tfrac12p^2+\tfrac12w^2q^2)=q^4-12q^2+2$. Replacing $Q\tilde hQ$ by $\alpha Q$ gives an operator
below $\tilde h$ whose low spectrum is that of a $22\times22$ matrix, and by Sylvester's law of inertia (the number of negative pivots in an $LDL^\top$ factorization of $M-\lambda I$, for that matrix $M$,
equals the number of eigenvalues of $M$ below $\lambda$), an exact-sign factorization in ball arithmetic certifies
its eigenvalue bounds:
\[
e_0\ge-0.1991865064,\quad e_1\ge 2.3704667252,\quad
o_0\ge-0.1034946187,\quad o_1\ge 3.5733502037 .
\]
The second level of $H_{\rm sep}$ in $\mathcal S$ is at least
$\min(o_1+2e_0,\ o_0+e_0+e_1)$. Hence
\[
\lambda_2(H|_{\mathcal S})\ \ge\ \rho=2.0677856001 .
\]
In the sector of functions even in every $q_i$ and symmetric under $q_2\leftrightarrow q_3$,
which contains the ground state, the same argument gives $\rho_0=2e_0+e_1\ge1.9720937124$.

\paragraph{Step 2: Temple enclosure.}
A Ritz vector $\tilde\psi\in\mathcal S$ (Hermite basis, 30 functions per axis in the relevant
parity, 13950 sector functions) is stored with one coefficient per symmetry orbit, so it lies
in $\mathcal S$. $H$ is applied to $\tilde\psi$: since $V$ raises the Hermite index by at most $4$ in each direction,
$H\tilde\psi$ and $(H-E_1)^2\tilde\psi$ (needed in Step~4) lie in the span of the first $34$ basis functions of that
parity per axis, where they are computed without truncation. With
$\eta=\ip{\tilde\psi}{H\tilde\psi}$ and $\varepsilon\ge\|(H-\eta)\tilde\psi\|$
($\|\tilde\psi\|=1$), Temple's inequality~\cite{Temple} and $\eta<\rho$ give
\[
\eta-\frac{\varepsilon^2}{\rho-\eta}\le E_1\le\eta,\qquad
\sin^2\theta\le\frac{\eta-E_1}{\rho-E_1},
\]
where $\theta$ is the angle between $\tilde\psi$ and $\psi$. Numerically
$\varepsilon\le 2.67\cdot10^{-9}$, which gives the enclosure of Theorem~\ref{thm:main} and
$\sin\theta\le5.9\cdot10^{-9}$.

\paragraph{Step 3: a graph-norm lemma.}
\begin{lemma}\label{lem:cube}
In $\A$,
\[
20\,(H+5)^3+2\cdot10^5-\sum_{|m|\le3}m^\dagger m\;=\;\sum_k g_k^\dagger g_k ,
\]
where the sum on the left runs over the 84 monomials $m=q^\alpha p^\beta$ of degree $\le 3$.
Consequently
$\sum_{|m|\le3}\|mu\|^2\le 20\,\ip{u}{T^3u}+2\cdot10^5\|u\|^2$ for every Schwartz function $u$.
\end{lemma}

The certificate is a Gram representation
$\sum_{a,b}G_{ab}v_a^\dagger v_b$ over the 259 elements $v=i^{|\beta|}q^\alpha p^\beta$ of
weighted degree $|\alpha|+2|\beta|\le6$. Time reversal makes $G$ real, and axis parity splits
it into eight blocks of sizes $20$ to $47$. The identity is exact over $\Q$, and every block is
positive definite by exact $LDL^\top$ (smallest pivot $1.26$). The float solution was rounded
and projected exactly onto the linear constraints, following  Peyrl and Parrilo~\cite{PP}.
Lemma~\ref{lem:cube} is an explicit instance of a standard relative-boundedness estimate: for globally
elliptic operators such as $H$, powers of $H$ dominate all monomials $q^\alpha p^\beta$ of lower weighted
degree~\cite{Shubin,HR82,Helffer84}. The general theory gives the existence of such constants; the
certificate supplies usable ones.

\paragraph{Step 4: the error in all degree-3 directions.}
Let $\varphi=\tilde\psi-\cos\theta\,\psi$ (phase chosen so that $\cos\theta\ge0$) and
$\delta=\tilde\psi-\psi$. Since $(H-E_1)\psi=0$ we have
$(H-E_1)^j\varphi=(H-E_1)^j\tilde\psi$. The spectral measure of $\varphi$ lies in
$[\rho,\infty)$, and there $(\lambda+c)/(\lambda-E_1)\le \kappa:=(\rho+c)/(\rho-E_1)\le15.61$. Hence
\[
\ip{\varphi}{T^3\varphi}\le\|T\varphi\|\,\|T^2\varphi\|\le \kappa^3 r_1r_2,\qquad
r_j\ge\|(H-E_1)^j\tilde\psi\|,
\]
with $r_1\le2.67\cdot10^{-9}$ and $r_2\le2.75\cdot10^{-6}$, both computed in ball arithmetic
with $E_1$ replaced by its enclosure. Since
$\delta=\varphi-(1-\cos\theta)\psi$, Lemma~\ref{lem:cube} gives
\[
\varepsilon_X^2:=\sum_{|m|\le3}\|m\delta\|^2\le 20\big(\ip{\varphi}{T^3\varphi}^{1/2}+(1-\cos\theta)(E_1+c)^{3/2}\big)^2
+2\cdot10^5\big(\sin\theta+1-\cos\theta\big)^2,
\]
so that $\varepsilon_X\le2.38\cdot10^{-5}$.

\paragraph{Step 5: the moment matrix.}
Let $X$ (respectively $\tilde X$) be the map $\mathbb{C}^{84}\to L^2$ sending the coefficient
vector $a$ to $\sum_m a_m\,m\psi$ (respectively $\sum_m a_m\,m\tilde\psi$). Then
$M_3(B)=X^\dagger X$, $M_3(\tilde B)=\tilde X^\dagger\tilde X$, and
$\|X-\tilde X\|\le\|X-\tilde X\|_F=\varepsilon_X$. Since $\psi$ and $\tilde\psi$ have definite parity in each axis, $\langle m_a\psi,m_b\psi\rangle=0$
unless $m_a$ and $m_b$ have the same parity class $\big((\alpha_i+\beta_i)\bmod 2\big)_i$, so both
matrices are block diagonal, with blocks of sizes $10$, $18$ (three), $4$ (three) and $8$; the off-block
zeros are exact for both. The matrix $M_3(\tilde B)$ is computed in
ball arithmetic from 190 moments of $\tilde\psi$, and an $LDL^\top$ inertia count certifies
$\lambda_{\min}(M_3(\tilde B))\ge t=1.46919\cdot10^{-5}$. Therefore
\[
\lambda_{\min}(M_3(B))=\sigma_{\min}(X)^2\ge(\sqrt t-\varepsilon_X)^2\ge1.4510\cdot10^{-5}.
\]

\paragraph{Step 6: the value $B(f)$.}
Each term of $f$ is $f_e\,(p^\beta)^\dagger p^\beta$ with $\beta=e/2$ a degree-3 monomial, and
$\big|\|p^\beta\psi\|^2-\|p^\beta\tilde\psi\|^2\big|\le\varepsilon_X
(2\|p^\beta\tilde\psi\|+\varepsilon_X)$. Hence
$|B(f)-\tilde B(f)|\le3.78\cdot10^{-4}$ with $\tilde B(f)=40.84483$, which gives
$\tau^*=B(f)/|\ell(f)|<10683.13$.

The ground-state case is identical, in the sector described at the end of Step~1 (the ground
state is unique, hence fully symmetric, and lies in that sector).  Table~\ref{tab:cert} collects the certified quantities for both states.

\begin{table}[h]\centering
\begin{tabular}{lcc}\toprule
 & ground state & first excited level\\\midrule
sector & even $q_1,q_2,q_3$; $q_2\!\leftrightarrow\!q_3$ & odd $q_1$, even $q_2,q_3$; $q_2\!\leftrightarrow\!q_3$\\
energy & $1.53482611032904478$ & $1.61488137077986920$\\
gap bound $\rho$ & $1.97209$ & $2.06778$\\
$\sin\theta$ & $\le6.3\cdot10^{-9}$ & $\le5.9\cdot10^{-9}$\\
$\varepsilon_X$ & $\le2.52\cdot10^{-5}$ & $\le2.38\cdot10^{-5}$\\
$\lambda_{\min}(M_3(\tilde B))$ & $\ge1.43258\cdot10^{-5}$ & $\ge1.46919\cdot10^{-5}$\\
$\lambda_{\min}(M_3(B))$ & $\ge1.4135\cdot10^{-5}$ & $\ge1.4510\cdot10^{-5}$\\
$\tilde B(f)$ & $38.53712$ & $40.84483$\\
$|B(f)-\tilde B(f)|$ & $\le3.81\cdot10^{-4}$ & $\le3.78\cdot10^{-4}$\\
$\tau^*$ & $10079.54$ & $10683.13$\\\bottomrule
\end{tabular}
\caption{Certified quantities (energies to $\pm2\cdot10^{-17}$; lower bounds rounded down, upper bounds up).}\label{tab:cert}
\end{table}

The precision of the energies is a by-product of the small Ritz residual: the Temple enclosure
has width of order $\varepsilon^2/(\rho-\eta)$. The proof uses only $\varepsilon_X\ll\sqrt t$, which
holds with a large margin ($\varepsilon_X/\sqrt t\approx6\cdot10^{-3}$).

\paragraph{Proof of Proposition~\ref{prop:spec}.}
The eight parity classes are $H$-invariant and span $L^2(\R^3)$, and permutations of the axes
map classes with the same number of odd directions onto each other. So it suffices to treat
the four classes $(+,+,+)$, $(-,+,+)$, $(-,-,+)$ and $(-,-,-)$. Steps~1 and~2 are used in each class,
now without the $q_2\leftrightarrow q_3$ symmetry.

\emph{Second eigenvalue.} $H-H_{\rm sep}\ge0$ (Step~1), and both operators preserve every parity class, so by min-max, $H\ge H_{\rm sep}$ gives
$\lambda_k(H|_{\rm class})\ge\lambda_k(H_{\rm sep}|_{\rm class})$ for every $k$. The
eigenfunctions of $H_{\rm sep}$ in the class are products of one-dimensional eigenfunctions of
$\tilde h$ of the prescribed parities, so its second level there is
$\sum_i a_0(\sigma_i)+\min_i\big(a_1(\sigma_i)-a_0(\sigma_i)\big)$, where $\sigma_i$ is the parity in
direction $i$, $a_k(+)=e_k$, $a_k(-)=o_k$, and $e_k,o_k$ are the certified bounds of Step~1. For the
four classes this gives $\lambda_2(H|_{\rm class})\ge\rho_{\rm class}$ with $\rho_{\rm class}=1.9720$,
$2.0677$, $2.1634$ and $3.3663$ respectively.

\emph{Lowest eigenvalue.} In each class a Ritz vector is built as in Step~2 (Hermite basis of the
class parities), and Temple's inequality~\cite{Temple}, with the class's own $\rho_{\rm class}$, encloses
the lowest eigenvalue of the class as listed. Every enclosure lies below $1.9720\le\rho_{\rm class}$.
Each class therefore has exactly one eigenvalue below $1.9720$, and it is simple within the class. \qed

\section{Discussion}\label{sec:discussion}

\paragraph{The momentum block.}
The construction rests on one structural fact (Lemmas~\ref{lem:rows} and~\ref{lem:MDelta}): the level-3
rows never reach the top-degree pure-momentum monomials, so the $10\times10$ block of $M_3$ indexed by the
cubic momentum monomials $p^\beta$ (the lower-right corner of $M_3$ when these monomials are listed last),
the \emph{momentum corner}, is invisible to them. 
This holds at every
level (Appendix~\ref{sec:corner}). At level~3 it has a sharper consequence (Proposition~\ref{prop:corner}): the
relaxation is strictly feasible at a given $E$ if and only if it is strictly feasible with the momentum corner
deleted, so the momentum corner carries no information about which energies admit a strictly feasible point. 
This complements the picture in~\cite{K26}, where the high momentum moments were identified as the
unbounded directions of the stationary relaxation: for the eigenstate relaxation, the top-degree momentum
moments are genuinely free (Proposition~\ref{prop:corner}). Theorem~\ref{thm:main} says more: with a
witness in the gap between the two cones, moving along the free direction reaches points realized
by no state. The special role of the momentum moments has also been observed in a different setting:
for the Kronig--Penney model, positivity constraints built from position-type moments determine the
band gaps but leave the minimum energy unconstrained, and a finite lower bound on it appears only once
higher powers of the momentum enter the positivity matrices~\cite{BBB22}.

\paragraph{Numerical consequence.}
The Schur-complement bound shows why the full semidefinite program can be badly conditioned near the
lowest accepted energy: the required $\tau$ can grow without bound as the reduced matrix $A$ approaches
singularity, so feasible points there may require very large momentum moments. Numerical
instability of bootstrap programs has been reported before~\cite{LZ}. For~\eqref{eq:H} at $K=3$ we computed
the lowest energy at which each version is feasible, with the rows and parity reduction of
Definition~\ref{def:boot}, via CVXPY:
\begin{center}
\begin{tabular}{lcc}\toprule
 & full $M_3$ & momentum corner deleted\\\midrule
CLARABEL & $1.0697$ & $1.0599$\\
MOSEK & $1.0992$ & $1.0624$\\\bottomrule
\end{tabular}
\end{center}
By Proposition~\ref{prop:corner}, deleting the momentum corner does not change the energies admitting
strictly feasible points. At $E=1.065$, which both solvers report infeasible for the full program, we add
$\tau\Delta_\ell$ to the reduced program's minimum-norm solution, with $\ell$ given by Gaussian moments
and $\tau=10^8$. Numerically, the resulting full moment matrix has smallest eigenvalue
$+6.6\cdot10^{-5}$. These computations indicate feasibility at $E=1.065$; the solver-reported thresholds
exceed this value by about $0.005$ and $0.034$, respectively.
These are floating-point computations, not certified ones, and none of Theorem~\ref{thm:main} rests on them:
its proof uses no semidefinite solver, and every certificate is checked in exact or ball arithmetic.

\paragraph{Scope of the constraint set.}
Theorem~\ref{thm:main} concerns the relaxation of Definition~\ref{def:boot}, which imposes the rows generated
by single monomials, as is standard~\cite{HHK,BH1,BDDJK}; Lemma~\ref{lem:rows} covers every such row, so the
theorem is complete as stated. Some implementations impose more: linear combinations of rows generated by
monomials of higher degree whose high-degree parts cancel, leaving a relation in $\A_{\le6}$ (equivalently,
auxiliary moments of degree above $6$ that do not enter the moment matrix). For generators of degree
$\le8$, the spurious point survives this strengthening: every combination
$\sum_X a_X\,X(H-E)+\sum_Y b_Y\,(H-E)Y$ that lies in $\A_{\le6}$ 
is annihilated by $\Delta$, for every $E$. Indeed, writing the combination as $W(H-E)+\sum_Y b_Y[H,Y]$ with
$W=\sum_X a_XX+\sum_Y b_YY$, the top-degree part of $W(H-E)$ is $W_{\rm top}V_4\ne0$, of degree $\deg W+4$, while
$\deg[H,Y]\le10$; so membership in $\A_{\le6}$ forces $\deg W\le6$. Then $EW\in\A_{\le6}$, and the conditions on the
coefficients $a_X,b_Y$ do not depend on $E$. For fixed admissible coefficients, $\Delta$ of the combination is
affine in $E$; exact rank computations over $\Q$ show that it vanishes at five distinct values of $E$, hence at
every $E$; since $B$ satisfies all such relations automatically, $L_\tau$ remains feasible and
spurious. The hierarchy in~\cite{ZK26} imposes stationarity for every polynomial generator whose reduced
commutator lies in the retained moment space, with no bound on the generator degree. The degree-$8$
verification above does not establish that stronger statement for the present construction.

\section{Reproducibility}\label{sec:repro}

All files are available as ancillary files with the arXiv version of this paper,
\url{https://arxiv.org/abs/2610.03362}; the \texttt{README} lists the commands.
\begin{itemize}
\item \texttt{certifyK.py} produces \texttt{certK3\_A2.txt} (and \texttt{certK4\_A2.txt}),
containing $f$, $\ell$ and the SOS Gram matrix, and checks Lemmas~\ref{lem:witness}--\ref{lem:MDelta}.
\item \texttt{sosT3.py} and \texttt{exactT3.py} produce and certify
\texttt{lemmaT3\_A2\_c5.txt} (Lemma~\ref{lem:cube}).
\item \texttt{rig3x.py 2 4 20 30 5 100} carries out Steps 1--6 for the excited level, and
\texttt{rig3x.py 2 4 20 30 5 000} for the ground state.
\item \texttt{rigclass.py 2 4 20 22 5 P}, with $P\in\{000,100,110,111\}$, certifies
Proposition~\ref{prop:spec}.
\item \texttt{wcut.py} (with \texttt{conelib.py}) runs the full-cone alternation.
\item \texttt{m3check.py} is an independent floating-point recomputation, written from the definitions
of this paper, of $E_1$, $E_0$, $\lambda_{\min}(M_3(\tilde B))$ and $\tilde B(f)$; it agrees with the values
of Section~\ref{sec:proofB} to 9--10 digits. \texttt{fdcheck.py} computes the low spectrum by fourth-order
finite differences on a cubic grid, a method sharing no code with the rest; after Richardson extrapolation it
reproduces the eight levels of Proposition~\ref{prop:spec}, with multiplicities $1,3,3,1$, to within
$3\cdot10^{-5}$, and finds no other eigenvalue below $3.8$.
\item \texttt{eband3.py} computes the lowest accepted energy with and without the momentum corner,
and \texttt{tailcheck3.py} carries out the direct check of Section~\ref{sec:discussion}.
\item \texttt{indep/} contains an independent re-verification of every exact ingredient:
Lemmas~\ref{lem:witness}--\ref{lem:MDelta}, the closure checks, Lemma~\ref{lem:cube} and the
$K=4$ certificate. It
uses a separately written Weyl-algebra implementation based on the standard-ordering symbol
product, cross-checked against explicit differential operators.
\end{itemize}

\section{Conclusion}\label{sec:conclusion}

For the coupled three-dimensional double well~\eqref{eq:H}, the level-3 eigenstate bootstrap admits, at
the first excited level and at the ground state, a feasible functional with strictly positive definite
moment matrix that is realized by no state (Theorem~\ref{thm:main}). The spurious direction is the
top-degree momentum block, which the eigenstate and stationarity rows never reach at any level
(Appendix~\ref{sec:corner}), and a nonnegative momentum sextic that is not a sum of squares separates the
resulting point. Non-tightness persists at level 4 (Appendix~\ref{rem:higher}), where strict positive
definiteness is no longer possible; whether interior spurious points exist there, relative to the forced
kernel, is left open.

\paragraph{Funding information}
The authors received no specific funding for this work.

\begin{appendix}
\numberwithin{equation}{section}

\section{The SOS certificate for the witness}\label{app:G}

Let $v(p)$ be the vector of the 15 monomials of degree 4 in $p$. The certificate of
Lemma~\ref{lem:witness}(1) is the exact identity
\[
|p|^2 f(p)=v(p)^{\!\top}G\,v(p)=\sum_{m,m'}G_{m,m'}\,m\,m' ,
\]
with $f$ as in Table~\ref{tab:fl} and $G$ the rational symmetric $15\times15$ matrix below, indexed by
pairs of degree-4 monomials. Only the 39 nonzero entries with $m\le m'$ are listed, and
$G_{m',m}=G_{m,m'}$ by symmetry. Exact $LDL^\top$ gives 15 positive pivots, the smallest
$\approx1.74\cdot10^{-5}$, so $G\succ0$.

{\small
\renewcommand{\arraystretch}{1.25}
\begin{center}
\begin{tabular}{@{}ll@{\qquad}ll@{}}\toprule
$m\cdot m'$ & $G_{m,m'}$ & $m\cdot m'$ & $G_{m,m'}$\\\midrule
$p_3^4\cdot p_3^4$ & $\tfrac{11906231}{135000000}$ & $p_1p_3^3\cdot p_1p_3^3$ & $\tfrac{298705151}{4050000000}$\\
$p_3^4\cdot p_2^2p_3^2$ & $-\tfrac{116990393}{8100000000}$ & $p_1p_3^3\cdot p_1p_2^2p_3$ & $\tfrac{155230591}{8100000000}$\\
$p_3^4\cdot p_2^4$ & $-\tfrac{198401119}{4500000000}$ & $p_1p_3^3\cdot p_1^3p_3$ & $-\tfrac{52260383}{562500000}$\\
$p_3^4\cdot p_1^2p_3^2$ & $-\tfrac{616453673}{8100000000}$ & $p_1p_2p_3^2\cdot p_1p_2p_3^2$ & $\tfrac{20375309}{4050000000}$\\
$p_3^4\cdot p_1^2p_2^2$ & $\tfrac{366704579}{4050000000}$ & $p_1p_2p_3^2\cdot p_1p_2^3$ & $-\tfrac{195587303}{8100000000}$\\
$p_3^4\cdot p_1^4$ & $-\tfrac{198401119}{4500000000}$ & $p_1p_2p_3^2\cdot p_1^3p_2$ & $\tfrac{155223067}{8100000000}$\\
$p_2p_3^3\cdot p_2p_3^3$ & $\tfrac{474222323}{4050000000}$ & $p_1p_2^2p_3\cdot p_1p_2^2p_3$ & $\tfrac{2546989}{506250000}$\\
$p_2p_3^3\cdot p_2^3p_3$ & $-\tfrac{104543761}{1125000000}$ & $p_1p_2^2p_3\cdot p_1^3p_3$ & $-\tfrac{97786231}{4050000000}$\\
$p_2p_3^3\cdot p_1^2p_2p_3$ & $-\tfrac{195594827}{8100000000}$ & $p_1p_2^3\cdot p_1p_2^3$ & $\tfrac{118548797}{1012500000}$\\
$p_2^2p_3^2\cdot p_2^2p_3^2$ & $\tfrac{241597423}{2250000000}$ & $p_1p_2^3\cdot p_1^3p_2$ & $-\tfrac{418106923}{4500000000}$\\
$p_2^2p_3^2\cdot p_2^4$ & $-\tfrac{123310801}{1620000000}$ & $p_1^2p_3^2\cdot p_1^2p_3^2$ & $\tfrac{241545043}{2250000000}$\\
$p_2^2p_3^2\cdot p_1^2p_3^2$ & $-\tfrac{86943277}{1620000000}$ & $p_1^2p_3^2\cdot p_1^2p_2^2$ & $-\tfrac{434701391}{8100000000}$\\
$p_2^2p_3^2\cdot p_1^2p_2^2$ & $-\tfrac{217394971}{4050000000}$ & $p_1^2p_3^2\cdot p_1^4$ & $-\tfrac{116928293}{8100000000}$\\
$p_2^2p_3^2\cdot p_1^4$ & $\tfrac{73339159}{810000000}$ & $p_1^2p_2p_3\cdot p_1^2p_2p_3$ & $\tfrac{40751437}{8100000000}$\\
$p_2^3p_3\cdot p_2^3p_3$ & $\tfrac{149382613}{2025000000}$ & $p_1^2p_2^2\cdot p_1^2p_2^2$ & $\tfrac{483151907}{4500000000}$\\
$p_2^3p_3\cdot p_1^2p_2p_3$ & $\tfrac{31049081}{1620000000}$ & $p_1^2p_2^2\cdot p_1^4$ & $-\tfrac{616472033}{8100000000}$\\
$p_2^4\cdot p_2^4$ & $\tfrac{119072651}{1350000000}$ & $p_1^3p_3\cdot p_1^3p_3$ & $\tfrac{59263639}{506250000}$\\
$p_2^4\cdot p_1^2p_3^2$ & $\tfrac{366694247}{4050000000}$ & $p_1^3p_2\cdot p_1^3p_2$ & $\tfrac{597370153}{8100000000}$\\
$p_2^4\cdot p_1^2p_2^2$ & $-\tfrac{23386447}{1620000000}$ & $p_1^4\cdot p_1^4$ & $\tfrac{119045273}{1350000000}$\\
$p_2^4\cdot p_1^4$ & $-\tfrac{24800141}{562500000}$ & &\\
\bottomrule
\end{tabular}
\end{center}}

\section{The momentum corner at every level}\label{sec:corner}

The construction of Section~\ref{sec:construction} uses one structural fact: the level-3 rows
never reach the top-degree pure-momentum monomials. This holds at every level, and it determines
what the top-degree momentum block of the moment matrix can contribute. For $K\ge2$ write
$\pi_{2K}(A)=\sum_{|\beta|=2K}c_{0\beta}\,p^\beta$ for the top-degree pure-momentum part of
$A\in\A_{\le2K}$, and call the monomials $p^\beta$, $|\beta|=K$, the \emph{momentum corner} of
$M_K$.

\begin{proposition}\label{prop:corner}
Let $K\ge2$, and let $\ell$ be any real linear functional on the forms of degree $2K$ in
$p_1,p_2,p_3$ whose Gram matrix $G_\ell=(\ell(p^{\beta+\beta'}))_{|\beta|=|\beta'|=K}$ is positive
definite (for instance the Gaussian moments $\ell(p^e)=\mathbb E\,\xi^e$, $\xi\sim N(0,I_3)$). Then
$\Delta_\ell=\ell\circ\pi_{2K}$ satisfies:
\begin{enumerate}
\item $\Delta_\ell(1)=0$, $\Delta_\ell$ is Hermitian, and $\Delta_\ell$ annihilates every eigenstate
and stationarity row of Definition~\ref{def:boot} at level $K$, for every $E$;
\item $M_K(\Delta_\ell)=G_\ell\oplus0$, supported on the momentum corner.
\end{enumerate}
Consequently, for $K\le3$ and every $E$: the level-$K$ relaxation is strictly feasible (a feasible $L$ with
$M_K(L)\succ0$) if and only if there is an $L_0$ satisfying conditions 1, 3 and 4 of
Definition~\ref{def:boot} whose moment matrix \emph{with the momentum corner deleted} is positive
definite. In that case $L_0+\tau\Delta_\ell$ is strictly feasible for all sufficiently large $\tau$.
\end{proposition}

\begin{proof}
(1) $\Delta_\ell(1)=0$ since $2K>0$. Hermiticity: $(q^\alpha p^\beta)^\dagger=q^\alpha p^\beta+{}$terms of
lower degree, so $\pi_{2K}(A^\dagger)$ has the complex-conjugate coefficients of $\pi_{2K}(A)$, and
$\ell$ is real. In the standard-ordered expansion, a term of $q^\alpha p^\beta\,V$ or of
$V q^\alpha p^\beta$ is free of $q$ only if $\alpha=0$ and all position factors of a monomial of $V$
are removed by derivatives or are absent (the constant term); such a term has degree at most
$\deg O$. Hence for an eigenstate row with $\deg O\le2K-4$, the pure-momentum part of
$O(H-E)$ and of $(H-E)O$ has degree at most $\deg O+2\le2K-2$. For a stationarity row write
$O=q^\alpha p^\beta$. If $\beta\ne0$, the standard-ordered expansion of $[V,O]=q^\alpha[V,p^\beta]$
contains, for each $i$ with $\beta_i\ge1$, the terms $i\beta_i\,q^\alpha(\partial_iV)_3\,p^{\beta-e_i}$ of
degree $\deg O+2$, where $(\partial_iV)_3=8q_i^3-2q_i\sum_{j\ne i}q_j^2\ne0$ is the cubic part of
$\partial_iV$. Distinct $i$ give distinct $p$-exponents and all other terms have lower degree, so
these cannot cancel, while $[\tfrac12|p|^2,O]$ has degree at most $\deg O$. Hence
$[H,O]\in\A_{\le2K}$ forces $\deg O\le2K-2$, and then every pure-momentum term of $[H,O]$ has degree
at most $\deg O\le 2K-2$. If $\beta=0$, then $[V,O]=0$ and the pure-momentum terms of
$[\tfrac12|p|^2,q^\alpha]$ have degree at most $1$. Every admissible row
therefore has $\pi_{2K}=0$.

(2) The product $m_a^\dagger m_b$ of two monomials of degree $\le K$ has a pure-momentum term of degree
$2K$ only if both are pure-momentum of degree $K$, and then it equals $p^{\beta_a+\beta_b}$.

For the consequence, write $M_K(L_0)=\begin{pmatrix}A&C\\C^\dagger&D\end{pmatrix}$ with $A$ the matrix
with the momentum corner deleted. If $A\succ0$, then
$M_K(L_0+\tau\Delta_\ell)=\begin{pmatrix}A&C\\C^\dagger&D+\tau G_\ell\end{pmatrix}$ has Schur complement
$D+\tau G_\ell-C^\dagger A^{-1}C\succ0$ as soon as $\tau\lambda_{\min}(G_\ell)>\lambda_{\max}(C^\dagger A^{-1}C-D)$,
and $L_0+\tau\Delta_\ell$ satisfies conditions 1, 3 and 4 by part (1). Conversely, $A$ is a
principal submatrix of $M_K(L)$.
\end{proof}

So for $K\le3$ the momentum corner of $M_K$ carries no information about which energies admit a strictly
feasible point. For $K\ge4$ parts (1) and (2) still hold, but no feasible functional has a positive definite
moment matrix (Appendix~\ref{rem:higher}), so the equivalence says nothing there.

\section{Higher levels}\label{rem:higher}

Two statements must be kept apart: non-tightness (a spurious point exists in the feasible set) and the
interior statement of Theorem~\ref{thm:main} (a spurious point whose moment matrix is strictly positive
definite).

\emph{Non-tightness persists.} The exact part (Lemmas~\ref{lem:witness}--\ref{lem:MDelta} with degree-$2K$ data) is
also certified at $K=4$ (\texttt{certK4\_A2.txt}, $\ell(f)\approx-1.2956\cdot10^{-3}$) and checked independently. It
was found by the same alternation, seeded with a genuine degree-8 witness (a seed built from the degree-6 one
collapses to the trivial value $0$); the numerical value there is $-1.6403\cdot10^{-3}$. Since $B(f)<\infty$ for every
eigenstate, this shows, with no numerical analysis, that the level-4 eigenstate relaxation is not tight; the same
holds at $K=3$.

\emph{Strict positive definiteness does not.} At $K\ge4$, the polynomial $H-E$ has degree at most $K$.
Taking $O=H-E$ in the eigenstate rows gives $L((H-E)^2)=0$; positive semidefiniteness then places its
coefficient vector in $\ker M_K$. Thus the strict positive definiteness asserted in
Theorem~\ref{thm:main} is impossible at these levels. This does not rule out interior relative to the
equality constraints. A natural next step at $K=4$ is to prove positivity on the complement of the
forced kernel, with a degree-4 analogue of Lemma~\ref{lem:cube} providing the needed moment control.

\end{appendix}

\end{document}